\documentclass[%
 aip,
 amsmath,amssymb,
 reprint,%
]{revtex4-1}

\usepackage[dvipdfmx]{graphicx}
\usepackage{dcolumn}
\usepackage{bm}
\usepackage{amsmath, amsthm, amsxtra, amsfonts, amssymb,amscd,mathtools}
\usepackage{latexsym}
\usepackage{ulem,color}

\usepackage[utf8]{inputenc}
\usepackage[T1]{fontenc}
\usepackage{mathptmx}
\usepackage{etoolbox}

\def\eps{{\varepsilon}}
\def\N{{\mathbf N}}

\newtheorem{theorem}{Theorem}
\newtheorem{lemma}{Lemma}
\newtheorem{remark}{Remark}
\newtheorem{definition}{Definition}
\theoremstyle{definition}

\makeatletter
\def\@email#1#2{%
 \endgroup
 \patchcmd{\titleblock@produce}
  {\frontmatter@RRAPformat}
  {\frontmatter@RRAPformat{\produce@RRAP{*#1\href{mailto:#2}{#2}}}\frontmatter@RRAPformat}
  {}{}
}%
\makeatother
\begin{document}

\title{Synchronization and stability analysis of an exponentially diverging solution in a mathematical model of asymmetrically interacting agents}

\author{Yusuke Kato}
\email[]{yuukato@g.ecc.u-tokyo.ac.jp}
\author{Hiroshi Kori}
\affiliation{
Department of Complexity Science and Engineering, Graduate School of Frontier Sciences, The University of Tokyo, Kashiwa, Chiba 277-8561, Japan
}

\date{\today}

\begin{abstract}
This study deals with an existing mathematical model of asymmetrically interacting agents. We analyze the following two previously unfocused features of the model: (i) synchronization of growth rates and (ii) initial value dependence of damped oscillation.
By applying the techniques of variable transformation and time-scale separation, we perform the stability analysis of a diverging solution. We find that (i) all growth rates synchronize to the same value that is as small as the smallest growth rate and (ii) oscillatory dynamics appear if the initial value of the slowest-growing agent is sufficiently small. Furthermore, our analytical method proposes a way to apply stability analysis to an exponentially diverging solution, which we believe is also a contribution of this study. 
Although the employed model is originally proposed as a model of infectious disease, we do not discuss its biological relevance but merely focus on the technical aspects.

\end{abstract}

\maketitle

\begin{quotation}
Nonlinear dynamical systems exhibit a wide variety of behaviors, e.g., hysteresis, limit cycle, synchronization, chaos, etc.  
Elucidating these nonlinear phenomena is an important issue, as well as the development of new analytical methods.
Here, we analyze an existing nonlinear system and clarify the two properties that were not previously focused on, i.e., synchronization of growth rates and initial value dependence of damped oscillation. In the analysis, we propose a novel method to perform the linear stability analysis of an exponentially diverging solution, which is expected to help further investigate the nonlinear dynamics. 
\end{quotation}

\section{Introduction}
It is widely known that nonlinear systems show various complex behaviors \cite{strogatz2018nonlinear}. Those behaviors are extensively studied since the late 19th century: the discovery of chaotic dynamics on a strange attractor \cite{lorenz1963deterministic} and the analysis of synchronization transition \cite{winfree1967biological,kuramoto84} are examples of such theoretical research. Nonlinear dynamical systems are also used to describe various phenomena in nature and society, especially in population dynamics \cite{murray2002mathematical}. For example, the outbreak of a certain insect is explained as hysteresis \cite{ludwig1978qualitative}, 
varying prey-predator populations are described as periodic orbits \cite{lotka1920undamped, volterra1928variations}, and synchronous firefly flashing \cite{mirollo1990synchronization} or frog calls \cite{aihara2011complex} are modeled with coupled nonlinear oscillators. 
Therefore, it is important to establish theoretical frameworks for nonlinear phenomena, which would provide deeper insight into complex phenomena and contribute to broadening applications.
In addition, since most of the nonlinear differential equations cannot be explicitly solved, the development of novel analytical techniques is crucial to further study the nonlinear system. 

In the present paper, we focus on peculiar dynamical behavior observed in the model proposed in Ref \cite{nowak1992}. 
The model is described by a $2n$-dimensional dynamical system. The model elements are divided into 2 groups of $n$ agents (i.e., $v_i$ and $x_i$ in Eq. \eqref{original_model}) and these groups interact with each other asymmetrically. 

In the original paper, the authors investigated the dynamics of this model in both analytical and numerical ways \cite{nowak1992}. However, there are several open questions in their study. First, they did not analyze the following two features of numerical results: the amount of agents in one group (i.e., $v_i$ in Eq. \eqref{original_model}) (i) initially oscillates and then decreases to a very low level and (ii) increases extremely slowly despite the existence of fast-growing agents. Second, the analysis in the original paper was valid only under the assumption that the dynamics of agents in the other group (i.e., $x_i$ in Eq. \eqref{original_model}) are sufficiently fast.


We are particularly concerned with the two features observed numerically because they are considered to reflect the nonlinearity of the model. In addition, we expect that removing the assumption of fast dynamics is necessary to analyze the initial oscillatory behavior. Therefore, in this study, 
we aim to clarify the mechanisms of initial oscillation and the extremely slow growth by analyzing this model under more general conditions; i.e., without assuming the fast dynamics.

The summary of our results is as follows: in the case when $n=2$, numerical simulations suggest that the initial oscillation and the following slow growth of agents in one group (i.e., $v_i$ in Eq. \eqref{original_model}) appear if one agent has a considerably lower growth rate than the other. We perform the existence and linear stability analysis without assuming that the dynamics of agents in the other group (i.e., $x_i$ in Eq. \eqref{original_model}) are sufficiently fast. We determine that an oscillation occurs if the initial value of the slow-growing agent is sufficiently small. Next, we generalize these results for the case when $n \ge 3$ and one agent has a considerably lower growth rate $\eps$ than the others. In particular, we prove that (i) all growth rates synchronize to the same value of $O(\eps)$ if the parameters 
satisfy a few conditions and (ii) damped oscillation exists if the initial value of the slowest-growing mutant is sufficiently small. 

Our work is a theoretical study that reveals nontrivial and previously unfocused features of a nonlinear dynamical system of asymmetrically interacting agents. This study is also novel in that we perform the stability analysis of the dynamics that oscillate and diverge, compared to the previous works that analyze the stability of equilibrium solutions in mathematical models regarding population dynamics \cite{murase2005stability,iwami2006frequency,iwami2008mathematical,liu1997nonlinear}. 

\section{Mathematical model and nondimensionalization}
Our model is based on the one proposed in Ref \cite{nowak1992}. The dynamics of $2n$ agents, which are originally introduced as $n$ viral mutants and corresponding immune cells \cite{nowak1992}, are described by the following dynamical system: 
\begin{subequations}
\label{original_model}
\begin{align}
\dot{v_i} &= \frac{d v_i}{dt} = v_i(r_i - p_i x_i), \\
\dot{x_i} &= \frac{d x_i}{dt} = k v_i - u x_i \sum_{j = 1}^{n} v_j , \label{original_model_x}
\end{align}
\end{subequations}
where $v_i$ denotes the amount of mutant virus $i$, $x_i$ is the quantity of strain-specific immune cells attacking the virus $i$, and $n$ represents the number of viral mutant strains ($1 \leq i \leq n$). The parameter $r_i$ is the growth rate of virus $i$, $p_i$ represents the strength of the immune attack on virus $i$, $k$ is the activation rate of immune cells, and $u$ represents the strength of the viral attack on immune cells. Note the asymmetric interaction between the virus and immune cells; even though each strain of immune cells $x_i$ is specific to virus $i$, virus $i$ can attack all strains of immune cells. The parameters $r_i, p_i, k,$ and $u$ are assumed to be positive constants. 

We introduce dimensionless quantities $\alpha_i \coloneqq \frac{r_i}{r_1},\ \tilde v_i \coloneqq \frac{u}{r_1}v_i,\ \tilde x_i \coloneqq \frac{p_i}{r_i}x_i,\ \tau \coloneqq r_1 t,$ and $q_i \coloneqq \frac{p_i k}{r_i u}$. By renaming $\tilde v_i \to v_i,\ \tilde x_i \to x_i$, and $\tau \to t$, we transform Eq. \eqref{original_model} into the following dimensionless system: 
\begin{subequations}
\label{nondim_model}
\begin{align}
\dot{v_i} &= \alpha_i v_i(1- x_i), \\
\dot{x_i} &= v_i\left( q_i - x_i\frac{\sum_{j = 1}^{n} v_j}{v_i} \right). \label{nondim_model_x}
\end{align}
\end{subequations}
The parameter $\alpha_i$ is the ratio of growth rates among viral mutants and $q_i$ represents the immunological strength of $x_i$ compared with the virulence of $v_i$. We assume $0<\alpha_n \leq \cdots \leq \alpha_2 \leq \alpha_1 =1$ without loss of generality. 

\section{The case when $n=1$}
When there is only one viral mutant (i.e., $n=1$), our model is given by
\begin{subequations}
\label{n1_model}
\begin{align}
\dot{v} &= v(1- x), \\
\dot{x} &= v\left( q - x \right).
\end{align}
\end{subequations}
Figure \ref{fig:nume1} presents the simulation results, where we assume that the initial value of $x$ is zero (i.e., $x(0)=0$) because virus-specific immunity has not been prepared at the beginning of infection. Then, the viral dynamics are classified into the following two types: (i) for $q<1$, the viral load continues to increase (Figs. \ref{fig:nume1} ({\bf a}) and ({\bf b})), and (ii) for $q>1$, the viral load initially increases and subsequently decreases, eventually converging to zero (Figs. \ref{fig:nume1} ({\bf c}) and ({\bf d})). We conclude that the initial oscillation and the following slow viral growth observed in the numerical simulation of the previous study \cite{nowak1992} cannot be reproduced in the case of one viral mutant.
\begin{figure*}
\centering
\includegraphics[width=\linewidth]{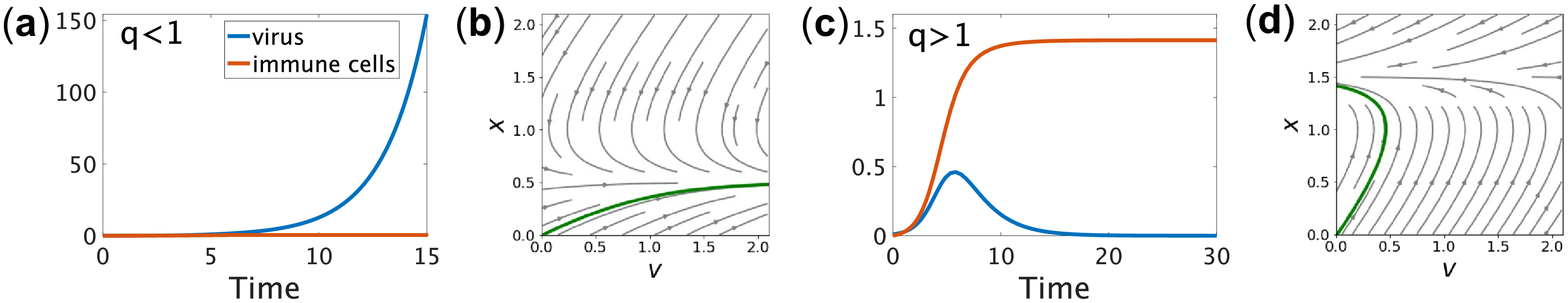}
\caption{The numerical simulations for one viral mutant case. Panels ({\bf a}) and ({\bf c}) represent the time course of the virus and immune cells, using the same initial conditions $v(0) = 0.01$ and $x(0) = 0$. Panels ({\bf b}) and ({\bf d}) display the phase planes where the green line in each phase plane is the trajectory that starts from $(v(0), x(0)) = (0.01, 0)$ . ({\bf a}, {\bf b}) When $q < 1$, the total viral load continues to increase (we set $q=0.5$). ({\bf c}, {\bf d}) When $q > 1$, the total viral load decreases after the initial peak and converges to zero (we set $q=1.5$). }
\label{fig:nume1}
\end{figure*}
\section{The case when $n=2$}
For $n=2$, our model is given as
\begin{subequations}
\label{n2_model}
\begin{align}
\dot v_1 &=  v_1(1 -  x_1), \label{n2_model_original_v1} \\
\dot v_2 &= \alpha_2 v_2(1 -  x_2), \label{n2_model_original_v2} \\
\dot x_1 &=  v_1\left[q_1 -  x_1\left(1+\frac{v_2}{v_1}\right)\right], \label{n2_model_x1}\\
\dot x_2 &=  v_2\left[q_2 -  x_2\left(1+\frac{v_1}{v_2}\right)\right]. \label{n2_model_x2}
\end{align}
\end{subequations}
\subsection{Simulation results}
Figure \ref{fig:n2} demonstrates the simulation results. In Figs. \ref{fig:n2} ({\bf a}) and ({\bf b}), we consider the situation in which two viral mutants have similar replication rates and the immunity is strong enough to eradicate the virus. Next, we weaken the immunity, or decrease $q_i$, so that the viral load diverges (Figs. \ref{fig:n2} ({\bf c}) and ({\bf d})). Finally, we notably reduce $\alpha_2$ to simulate a slow-replicating mutant (Figs. \ref{fig:n2} ({\bf e}) and ({\bf f})). Throughout these simulations, we use the same initial conditions ($v_1(0) = v_2(0) = 0.01$ and $x_1(0) = x_2(0) = 0$) because it is natural to assume that the amount of virus is very low and virus-specific immunity has not been established at the beginning of the infection. 
\begin{figure*}[ht]
\centering
\includegraphics[width=\linewidth]{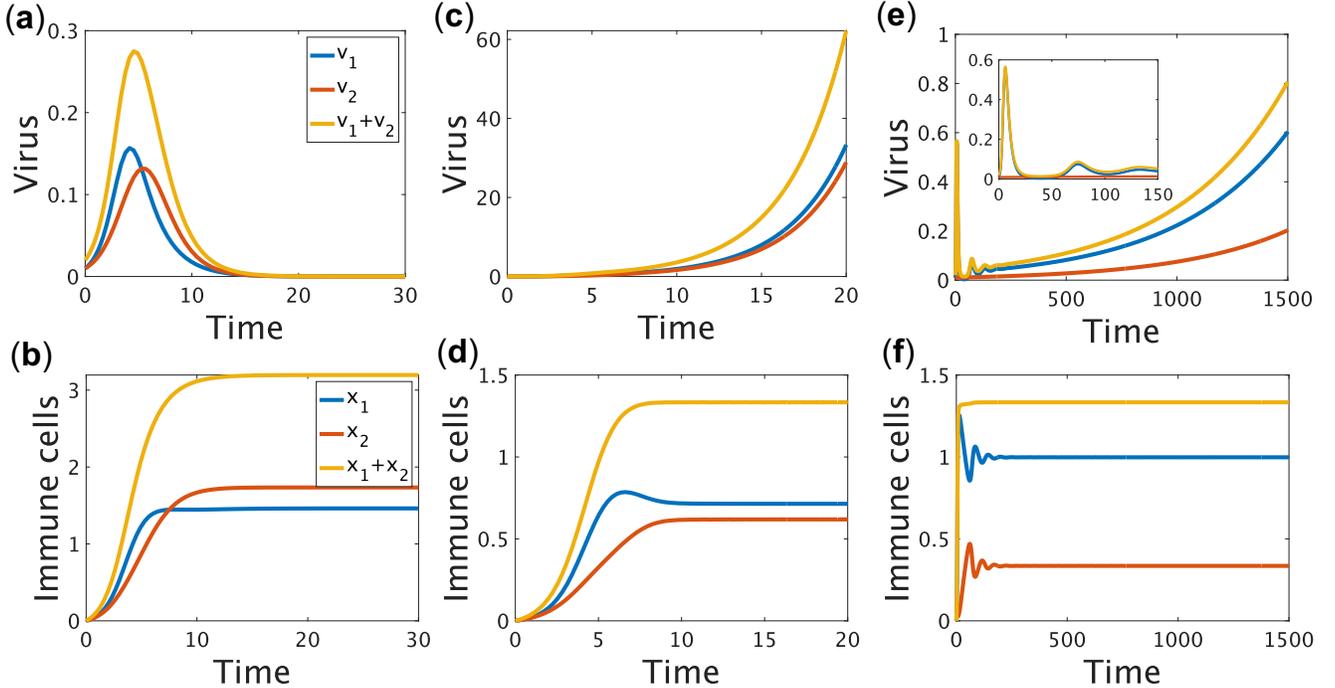}
\caption{Time course of the viral load ($v_1, v_2$) and immune cells ($x_1, x_2$) in various parameters. ({\bf a}, {\bf b}) Virus replicates transiently and is then eliminated from the body. ({\bf c}, {\bf d}) Virus continues to grow exponentially. ({\bf e}, {\bf f}) After initial proliferation, the viral load decreases to a low level. However, the viral load increases again and finally diverges. The parameters are as follows: $\alpha_2 = 0.75, q_1 = q_2 = 4$ for panels ({\bf a}) and ({\bf b}), $\alpha_2 = 0.75, q_1 = q_2 = \frac{4}{3}$ for panels ({\bf c}) and ({\bf d}), and $\alpha_2 = 0.003, q_1 = q_2 = \frac{4}{3}$ for panels ({\bf e}) and ({\bf f}) . We use the same initial conditions $v_1(0) = v_2(0) = 0.01$ and $x_1(0) = x_2(0) = 0$.}
\label{fig:n2}
\end{figure*}

The dynamics observed in Figs. \ref{fig:n2} ({\bf a}), ({\bf b}), ({\bf c}), and ({\bf d}) are qualitatively the same as those obtained in the previous section. In contrast, Figs. \ref{fig:n2} ({\bf e}) and ({\bf f}) demonstrate a new pattern with two components, namely, damped oscillation and slow exponential viral growth, which was previously observed but not analyzed \cite{nowak1992}. We are going to clarify the origin of this simulation result. 
\subsection{Analysis}
We expect that the dynamics in Figs. \ref{fig:n2} ({\bf e}) and ({\bf f}) may arise when $\alpha_2 \ll 1$. Thus, we treat $\alpha_2$ as a small parameter and put $\eps \coloneqq \alpha_2$. The other parameters are assumed to be $O(1)$. 

In Figs. \ref{fig:n2} ({\bf e}) and ({\bf f}), we observe that $x_i(t)$ converges toward a nonzero constant, denoted by $x_i^*$, whereas $v_i$ diverges. Based on Eq. \eqref{n2_model}, this condition is only possible when $\beta \coloneqq \frac{v_2}{v_1}$ also converges toward a positive constant, represented by $\beta^*$. Assuming the convergence of $\beta$ to $\beta^*$, we obtain the fixed point $x_i^*$ in Eqs. \eqref{n2_model_x1} and \eqref{n2_model_x2}, which is given as
\begin{gather}
    x_1^* = \frac{q_1}{1+\beta^*}, \quad {\rm and} \quad x_2^* = \frac{ q_2 \beta^*}{1+\beta^*}. \label{dif_x_i^*_n2}
\end{gather}
Substituting $x_i=x_i^*$ into $\dot \beta = \frac{v_1 \dot v_2 - v_2 \dot v_1}{v_1^2} = 0$, we further obtain
\begin{align}
    \beta^*=\frac{q_1-1+\eps}{1+\eps(q_2-1)}    = q_1-1 + O(\eps). \label{al_2}
\end{align}
For sufficiently small $\eps$, the condition $\beta^* > 0$ holds when $q_1 > 1$; thus, we assume this inequality below. Substituting $x_i = x_i^*$ and $\beta = \beta^*$ into Eqs. (\ref{n2_model_original_v1}) and (\ref{n2_model_original_v2}), we obtain 
\begin{align*}
 \dot v_1 =  \lambda v_1, \quad {\rm and} \quad \dot v_2 =  \lambda v_2,
\end{align*}
where 
\begin{align}
\lambda = \frac{\eps(q_1+q_2-q_1 q_2)}{q_1+\eps q_2} = O(\eps). \label{lam_2}
\end{align}
We also assume $q_1+q_2-q_1 q_2 > 0$ so that $\lambda > 0$. Therefore, if $x_i$ and $\beta$ sufficiently approach $x_i^*$ and $\beta^*$, respectively, $v_1$ and $v_2$ exponentially increase with the same time scale $\lambda^{-1}$ of $O(\eps^{-1})$. In other words, the effective replication rates of both viral mutants synchronize to the same value $\lambda$, which is as small as the slow-replicating mutant's replication rate $\eps$. 

We now perform the stability analysis of the obtained solution by invoking the notion of time-scale separation. For convenience, we introduce new variables $w_i(t)$ as $w_i \coloneqq v_i e^{-\lambda t}$. Then, Eq. \eqref{n2_model} is transformed to the following four-dimensional nonautonomous system:
\begin{subequations}
\label{n2_model2}
\begin{align}
 \dot w_1 &=  w_1(1 - \lambda- x_1), \\
 \dot w_2 &= \eps w_2(1 - \frac{\lambda}{\eps} -x_2), \label{n2_model2_w2}\\
\dot x_1 &=  e^{\lambda t} w_1\left[q_1 -  x_1\left(1+\frac{w_2}{w_1}\right)\right], \label{n2_model2_x1}\\
 \dot x_2 &=  e^{\lambda t} w_2\left[q_2 -  x_2\left(1+\frac{w_1}{w_2}\right)\right]. \label{n2_model2_x2}
\end{align}
\end{subequations}
As far as $t=O(1)$, we can safely replace $e^{\lambda t}$ in Eqs. (\ref{n2_model2_x1}) and (\ref{n2_model2_x2}) with $1$ because $e^{\lambda t} = 1 + O(\lambda t) = 1 + O(\eps)$. Moreover, $w_2$ is a slow variable. Thus, in a good approximation, the dynamics of $w_1, x_1$, and $x_2$ are described by the following three-dimensional autonomous subsystem: 
\begin{subequations}
\label{n2_model3}
\begin{align}
 \dot w_1 &=  w_1(1 - \lambda- x_1), \\
\dot x_1 &=  w_1\left[q_1 -  x_1\left(1+\frac{w_2}{w_1}\right)\right], \label{n2_model3_x1}\\
 \dot x_2 &=  w_2\left[q_2 -  x_2\left(1+\frac{w_1}{w_2}\right)\right], \label{n2_model3_x2}
\end{align}
\end{subequations}
in which $w_2$ is regarded as a constant. This subsystem has a nontrivial fixed point $(w_1,x_1,x_2)=(w_1^*,x_1^*,x_2^*)$, where $w_1^* = \frac{w_2}{\beta^*}$ and $x_i^*$ are given in Eq. \eqref{dif_x_i^*_n2}. The Jacobian matrix at this fixed point is 
\begin{equation}
\label{jacob_2}
\begin{pmatrix}
0 & -w_1^* & 0 \\
q_1 - x_1^* &  -(w_1^* + w_2) & 0 \\
-x_2^* &  0 & -(w_1^* + w_2)
\end{pmatrix}
,
\end{equation}
and its eigenvalues are
\begin{gather}
\label{eigen_2}
-\frac{(1+\beta^*)w_2}{\beta^*} \quad {\rm and} \quad \frac{-\frac{(1+\beta^*)w_2}{\beta^*} \pm \sqrt{D_0}}{2}, 
\end{gather}
where
\begin{equation}
\label{D_2}
D_0 = \frac{(1+\beta^*)^2 w_2^2}{(\beta^*)^2} - \frac{4 q_1 w_2}{1+\beta^*}.
\end{equation}
Regardless of the sign of $D_0$, all eigenvalues have negative real parts. Thus, the fixed point under consideration is asymptotically stable. We also determine that imaginary eigenvalues appear if $D_0<0$; i.e.,
\begin{equation}
0 < w_2 < \frac{4 q_1 (\beta^*)^2}{(1+\beta^*)^3} = \frac{4 q_1 (q_1 - 1 + \eps)^2 (1+\eps(q_2 -1))}{(q_1 + \eps q_2)^3}. \label{nece_2}
\end{equation}
Oscillation arises in this case. 

The fast variables stay in the $\eps$-vicinity of the fixed point in the full system after the transient process because the subsystem of the fast variables has a stable fixed point $(w_1,x_1,x_2)=(w_1^*,x_1^*,x_2^*)$. Substituting $x_2=x_2^*+O(\eps)$ into Eq. \eqref{n2_model2_w2} and further using Eqs. (\ref{al_2}) and (\ref{lam_2}), we obtain $\dot w_2= O(\eps^2)$, which implies that $w_2(t) = w_2(0)+O(\eps^2)$ for $t = O(1)$. Therefore, $w_2$ in inequality (\ref{nece_2}) can be regarded as $w_2(0)$ in a good approximation. Consequently, we conclude that oscillation inevitably occurs if there is a viral mutant whose replication rate is considerably smaller than the other's and its initial value is sufficiently small; i.e., 
\begin{equation}
v_2(0)=w_2(0)< \frac{4 q_1 (\beta^*)^2}{(1+\beta^*)^3}.
\end{equation}
Moreover, both viral mutants share an effective growth rate of $\lambda=O(\eps)$, namely, the slow mutant entrains the fast mutant. This synchronization underlies the emergence of the phase of low viral load.

By applying the same analysis for the case of three viral mutants (i.e., $n=3$), we also find that oscillatory viral dynamics and synchronized replication rates are observed if one viral mutant has a considerably lower replication rate than the others and its initial value is sufficiently small. See Appendix for more details about the three viral mutants case.

We generalize these results for $n$ mutants case in which one mutant has a considerably lower replication rate than the others. 
\section{General case}
We consider the system (\ref{nondim_model}) for the general case of $n$ mutants. Let $\Lambda$ be a real number. By introducing new variables $w_i \coloneqq v_i e^{-\Lambda t}$, we transform Eq. \eqref{nondim_model} into the following $2n$-dimensional nonautonomous system:
\begin{subequations}
\label{n_model}
\begin{align}
\dot{w_i} &= \alpha_i w_i(1- \frac{\Lambda}{\alpha_i} -x_i), \label{n_model_w}\\
\dot{x_i} &= e^{\Lambda t} w_i\left( q_i - x_i\frac{\sum_{l = 1}^{n} w_l}{w_i} \right),  \label{n_model_x}
\end{align}
\end{subequations}
for $1 \leq i \leq n$. As in the case when $n=2$, we are particularly concerned with the situation in which one of the agents $v_i$ has a considerably lower growth rate than the others. Thus, we treat $\alpha_n$ as a small parameter and put $\eps \coloneqq \alpha_n$. The other parameters are assumed to be $O(1)$. 
\begin{definition}
We define an internal fixed point as the fixed point whose coordinates are all positive. 
\end{definition}
First, we determine $\Lambda$ so that the system (\ref{n_model}) has an internal fixed point. 
\begin{theorem}
\label{theo_iff}
The system (\ref{n_model}) has at least one internal fixed point if and only if the following two conditions hold: 
\begin{equation}
 \Lambda = \frac{\left(\sum_{l=1}^{n}\frac{1}{q_l}\right)-1}{\sum_{l=1}^{n}\frac{1}{\alpha_l q_l}}, 
 \label{lam_n}
\end{equation}
\begin{equation}
\eps \sum_{l=1}^{n}  \frac{1}{\alpha_l q_l} > \left(\sum_{l=1}^{n}  \frac{1}{q_l}\right) -1.
\label{til_x_necessary_sup}
\end{equation}
\end{theorem}
\begin{proof}[Proof of necessity]
Let $(w_i^{\dagger}, x_i^{\dagger})$ with $w_i^{\dagger}> 0$ and $x_i^{\dagger}>0$ be the internal fixed point of system (\ref{n_model}). Then, $(w_i^{\dagger}, x_i^{\dagger})$ satisfies
\begin{equation}
x_i^{\dagger} = 1-\frac{\Lambda}{\alpha_i}, \label{dif_til_w} 
\end{equation}
and
\begin{equation}
w_i^{\dagger} q_i - x_i^{\dagger} \sum_{l = 1}^{n} w_l^{\dagger} = 0. \label{til_x}
\end{equation}
We rewrite Eq. \eqref{til_x} as 
\begin{equation}
\label{coefficent_equation}
Q_n
\pmb{w^{\dagger}}
= \pmb{0}_n, 
\end{equation}
where $Q_n$ is the coefficient matrix given as 
\begin{equation}
\label{Qn_def}
Q_n \coloneqq 
\begin{pmatrix}
q_1 - x_1^{\dagger} & -x_1^{\dagger} & -x_1^{\dagger} & \cdots & -x_1^{\dagger} \\
-x_2^{\dagger} & q_2 - x_2^{\dagger} & -x_2^{\dagger} & \cdots & -x_2^{\dagger} \\
\vdots & \vdots & \vdots & & \vdots \\
-x_n^{\dagger} & -x_n^{\dagger} & -x_n^{\dagger} & \cdots & q_n - x_n^{\dagger} 
\end{pmatrix}
,
\end{equation}
$\pmb{w^{\dagger}}\coloneqq{}^{t}(w_1^{\dagger}, w_2^{\dagger}, \ldots, w_n^{\dagger})$, and $\pmb{0}_n$ denotes the zero vector of order $n$. Since we assume $\pmb{w^{\dagger}} \neq \pmb{0}_n$, Eq. \eqref{coefficent_equation} has a nontrivial solution; i.e., 
\begin{equation}
\label{Qn_0}
\det Q_n = 0.
\end{equation}
\begin{lemma}
\label{Qn_calcu}
The determinant of the matrix $Q_n$ given in Eq. \eqref{Qn_def} is calculated as
\begin{equation}
\label{Delta_n}
\det Q_n= \left(1- \sum_{l=1}^{n}  \frac{x_l^{\dagger}}{q_l} \right) \prod_{k=1}^{n} q_k. 
\end{equation}
\end{lemma}
\begin{proof}[Proof of Lemma. \ref{Qn_calcu}.]
We prove this by induction on $n$. Obviously, Eq. \eqref{Delta_n} holds when $n=2$. We assume that Eq. \eqref{Delta_n} is true when $n=m$. Then, 
\begin{widetext}
\begin{align*}
\det Q_{m+1} &= 
\begin{vmatrix}
q_1 - x_1^{\dagger} & -x_1^{\dagger} &  \cdots & -x_1^{\dagger} & -x_1^{\dagger} \\
-x_2^{\dagger} & q_2 - x_2^{\dagger} &  \cdots & -x_2^{\dagger} & -x_2^{\dagger} \\
\vdots & \vdots &  & \vdots & \vdots \\
-x_m^{\dagger} & -x_m^{\dagger} &  \cdots & q_m -x_m^{\dagger} & - x_m^{\dagger} \\
-x_{m+1}^{\dagger} & -x_{m+1}^{\dagger} &  \cdots & -x_{m+1}^{\dagger} & q_{m+1} - x_{m+1}^{\dagger} 
\end{vmatrix} \\
&=
\begin{vmatrix}
q_1 - x_1^{\dagger} & -x_1^{\dagger} &  \cdots & -x_1^{\dagger} & 0 \\
-x_2^{\dagger} & q_2 - x_2^{\dagger} &  \cdots & -x_2^{\dagger} & 0 \\
\vdots & \vdots &  & \vdots & \vdots \\
-x_m^{\dagger} & -x_m^{\dagger} &  \cdots & q_m -x_m^{\dagger} & 0 \\
-x_{m+1}^{\dagger} & -x_{m+1}^{\dagger} &  \cdots & -x_{m+1}^{\dagger} & q_{m+1} 
\end{vmatrix} 
+
\begin{vmatrix}
q_1 - x_1^{\dagger} & -x_1^{\dagger} &  \cdots & -x_1^{\dagger} & -x_1^{\dagger} \\
-x_2^{\dagger} & q_2 - x_2^{\dagger} &  \cdots & -x_2^{\dagger} & -x_2^{\dagger} \\
\vdots & \vdots &  & \vdots & \vdots \\
-x_m^{\dagger} & -x_m^{\dagger} &  \cdots & q_m -x_m^{\dagger} & - x_m^{\dagger} \\
-x_{m+1}^{\dagger} & -x_{m+1}^{\dagger} &  \cdots & -x_{m+1}^{\dagger} & - x_{m+1}^{\dagger} 
\end{vmatrix} \\
&= q_{m+1} \det Q_m + 
\begin{vmatrix}
q_1  & 0 & \cdots & 0 & -x_1^{\dagger} \\
 & q_2  & \cdots  & 0 & -x_2^{\dagger} \\
 &  & \ddots & \vdots & \vdots \\
 &  &   & q_m  & -x_m^{\dagger} \\
\text{\huge{0}} &  &   &  & - x_{m+1}^{\dagger} 
\end{vmatrix} \\
&= q_{m+1} \left(1- \sum_{l=1}^{m}  \frac{x_l^{\dagger}}{q_l} \right) \left(\prod_{k=1}^{m} q_k \right) - q_1 q_2 \cdots q_m x_{m+1}^{\dagger} = \left(1- \sum_{l=1}^{m+1}  \frac{x_l^{\dagger}}{q_l} \right) \prod_{k=1}^{m+1} q_k.
\end{align*}
\end{widetext}
\end{proof}
It thus follows from Eq. \eqref{Qn_0} and Lemma. \ref{Qn_calcu} that 
\begin{equation}
1- \sum_{l=1}^{n}  \frac{x_l^{\dagger}}{q_l} = 0. \label{lam_condition}
\end{equation}
By substituting Eq. \eqref{dif_til_w} into Eq. \eqref{lam_condition}, we obtain Eq. \eqref{lam_n}. Moreover, since we assume $x_n^{\dagger} = 1-\frac{\Lambda}{\eps}> 0$, the inequality (\ref{til_x_necessary_sup}) follows from Eq. \eqref{lam_n}. \\
{\it Proof of sufficiency.} We set $x_i^{\dagger} \coloneqq 1-\frac{\Lambda}{\alpha_i}$. Since $\det Q_n = 0$ follows from Eq. \eqref{lam_n}, Eq. \eqref {coefficent_equation} has a nontrivial solution, denoted by $\pmb{w^{\dagger}} \coloneqq {}^{t}(w_1^{\dagger}, w_2^{\dagger},\ldots,w_n^{\dagger}) \neq \pmb{0}_n$. Note that $w_i^{\dagger}$ satisfy Eq. \eqref{til_x}. Recall the inequality $\alpha_1 \geq \alpha_2 \geq \cdots \geq \alpha_{n-1} \gg \alpha_n$; we have $x_1^{\dagger} \geq x_2^{\dagger} \geq \cdots \geq x_{n-1}^{\dagger} > x_n^{\dagger}$ from $x_i^{\dagger} \coloneqq 1-\frac{\Lambda}{\alpha_i}$. Hence, the inequality (\ref{til_x_necessary_sup}), which is equivalent to $x_n^{\dagger}>0$, implies that $x_i^{\dagger} > 0$ for all $i$. It thus follows from Eq. \eqref{til_x} and $q_i>0$ that each $w_i^{\dagger}$ has the same sign with $\sum_{l = 1}^{n} w_l^{\dagger}$. Then, all $w_i^{\dagger}$ have the same sign, i.e, $w_i^{\dagger} > 0$ for all $i$ or $w_i^{\dagger} < 0$ for all $i$. If $w_i^{\dagger} < 0$ for all $i$, then $(x_i^{\dagger}, -w_i^{\dagger})$ is an internal fixed point of system (\ref{n_model}) because $-\pmb{w^{\dagger}} = {}^{t}(-w_1^{\dagger}, -w_2^{\dagger},\ldots,-w_n^{\dagger})$ is also a nontrivial solution of Eq. \eqref {coefficent_equation}. Hence, the system (\ref{n_model}) has at least one internal fixed point.

Thus, we have shown Theorem. \ref{theo_iff}.
\end{proof}

In the following discussion, we set $\Lambda$ as Eq. \eqref{lam_n} and assume the inequality (\ref{til_x_necessary_sup}). We also assume 
\begin{equation}
 \sum_{l=1}^{n}\frac{1}{q_l} > 1, \label{sum_n}
\end{equation}
so that $\Lambda > 0$. It follows from Eq. \eqref{lam_n} and the assumption $\eps \ll 1$ that
\begin{align}
\Lambda = \frac{\eps \left[\left(\sum_{l=1}^{n}\frac{1}{q_l}\right)-1\right]}{\frac{1}{q_n} + \eps \sum_{l=1}^{n-1}\frac{1}{\alpha_l q_l}} 
= O(\eps). \label{assump_lam_n}
\end{align}

Next, we perform the stability analysis assuming time-scale separation. As far as $t = O(1)$, we can safely replace $e^{\Lambda t}$ in Eq. \eqref{n_model_x} with $1$ because $e^{\Lambda t} = 1 + O(\Lambda t) = 1 + O(\eps)$. We also see that $w_n$ is a slow variable and $w_1, \ldots, w_{n-1}, x_1, \ldots, x_{n-1}$, and $x_n$ are fast variables. Thus, in a good approximation, the dynamics of these fast variables are described by the following $(2n-1)$-dimensional autonomous subsystem: 
\begin{subequations}
\label{n_model_2}
\begin{align}
\dot{w_j} &= \alpha_j w_j(1- \frac{\Lambda}{\alpha_j} -x_j), \label{n_model_2_w}\\
\dot{x_i} &= w_i\left( q_i - x_i\frac{\sum_{l = 1}^{n} w_l}{w_i} \right),  \label{n_model_2_x}
\end{align}
\end{subequations}
for $1\leq j \leq n-1$ and $1\leq i \leq n$. Note that $w_n$ is regarded as a constant. 

\begin{theorem}
The subsystem (\ref{n_model_2}) has the unique internal fixed point $(w_j^*, x_i^*)$ whose coordinates are given as
\begin{equation}
\label{dif_x_i^*}
x_i^* = 1- \frac{\Lambda}{\alpha_i}, 
\end{equation}
for $1\leq i \leq n$ and
\begin{equation}
\label{dif_w_j^*}
\begin{pmatrix}
w_1^* \\
w_2^* \\
\vdots \\
w_{n-1}^*
\end{pmatrix}
 = w_n R^{-1} 
\begin{pmatrix}
x_1^* \\
x_2^* \\
\vdots \\
x_{n-1}^*
\end{pmatrix}
, 
\end{equation}
where $R$ is a regular matrix written as 
\begin{equation}
\label{C_dif}
R \coloneqq 
\begin{pmatrix}
q_1 - x_1^* & -x_1^* & \cdots & -x_1^* \\
-x_2^* & q_2 -x_2^* & \cdots & -x_2^* \\
\vdots & \vdots & & \vdots \\
-x_{n-1}^* & -x_{n-1}^* & \cdots & q_{n-1} -x_{n-1}^* 
\end{pmatrix}
.
\end{equation}
\end{theorem}

\begin{proof}
Let us assume the existence of an internal fixed point $(w_j^*, x_i^*)$ of the subsystem (\ref{n_model_2}). Then, 
\begin{equation}
\label{x_i_dif_sup}
x_j^* = 1- \frac{\Lambda}{\alpha_j},
\end{equation}
\begin{equation}
\label{wj_dif}
w_j^* q_j - x_j^* \left(w_n + \sum_{l = 1}^{n-1} w_l^*\right) = 0,
\end{equation}
for $1\leq j \leq n-1$ and 
\begin{equation}
\label{wn_dif}
w_n q_n - x_n^* \left(w_n + \sum_{l = 1}^{n-1} w_l^*\right) = 0.
\end{equation}
Eq. \eqref{wj_dif} can be rewritten as
\begin{equation}
\label{coefficent_equation_2}
R
\begin{pmatrix}
w_1^* \\
w_2^* \\
\vdots \\
w_{n-1}^*
\end{pmatrix}
= w_n
\begin{pmatrix}
x_1^* \\
x_2^* \\
\vdots \\
x_{n-1}^*
\end{pmatrix}
,
\end{equation}
where $R$ is given in Eq. \eqref{C_dif}. According to Lemma \ref{Qn_calcu}, 
\begin{equation}
\label{detR}
\det R
= \left(1- \sum_{l=1}^{n-1}  \frac{x_l^*}{q_l} \right) \prod_{k=1}^{n-1} q_k.
\end{equation}
Note that 
\begin{equation}
\label{sumj_n-1}
1- \sum_{l=1}^{n-1}  \frac{x_l^*}{q_l}
= \frac{1}{\eps q_n \sum_{l=1}^{n}\frac{1}{\alpha_l q_l}} \left[ \eps \sum_{l=1}^{n}\frac{1}{\alpha_l q_l} - \left( \sum_{l=1}^{n}\frac{1}{q_l} \right) + 1 \right] >0
\end{equation}
follows from Eqs. \eqref{lam_n}, \eqref{x_i_dif_sup}, and inequality \eqref{til_x_necessary_sup}. 
Thus, $R$ is a regular matrix, which implies that Eq. \eqref{coefficent_equation_2} can be solved as 
\begin{equation}
\label{coefficent_equation_3}
\begin{pmatrix}
w_1^* \\
w_2^* \\
\vdots \\
w_{n-1}^*
\end{pmatrix}
= w_n R^{-1}
\begin{pmatrix}
x_1^* \\
x_2^* \\
\vdots \\
x_{n-1}^*
\end{pmatrix}
.
\end{equation}
By dividing both sides of Eq. \eqref{wj_dif} by $q_j$ and summing from $j=1$ to $n-1$, we obtain
\begin{align}
\sum_{l = 1}^{n-1} w_l^*  = \frac{w_n \sum_{j=1}^{n-1} \frac{x_j^*}{q_j}}{1- \sum_{j=1}^{n-1} \frac{x_j^*}{q_j}}. \label{sum_wj_sup}
\end{align}
Substituting Eq. \eqref{sum_wj_sup} into Eq. \eqref{wn_dif}, we have 
\begin{equation}
x_n^* = q_n\left( 1- \sum_{j=1}^{n-1} \frac{x_j^*}{q_j} \right) = 1-\frac{\Lambda}{\alpha_n}.
\end{equation}
Obviously, 
\begin{equation}
\label{x_i_positive}
x_1^* \geq x_2^* \geq \cdots \geq x_{n-1}^* > x_n^*>0, 
\end{equation}
from $\alpha_1 \geq \alpha_2 \geq \cdots \geq \alpha_{n-1} \gg \alpha_n$ and inequality (\ref{sumj_n-1}). It also follows from inequalities \eqref{sumj_n-1}, \eqref{x_i_positive}, and Eq. \eqref{sum_wj_sup} that $\sum_{l = 1}^{n-1} w_l^* > 0$, which implies that $w_j^* > 0$ for $1 \leq j \leq n-1$ from Eq. \eqref{wj_dif}. Therefore, we have shown that the point $(w_j^*, x_i^*)$ given by Eqs. \eqref{dif_x_i^*} and \eqref{dif_w_j^*} is the unique internal fixed point of the subsystem (\ref{n_model_2}). 
\end{proof}
\begin{remark}
Obviously, $x_i^*$ and $\frac{w_j^*}{w_n}$ are independent of $w_n$. 
\end{remark}
Let $J$ be the Jacobian matrix at the fixed point ($w_j^*,x_i^*$) of system (\ref{n_model_2}). Then, 
\begin{equation}
J =
\begin{pmatrix}
O_{n-1,n-1} & A & \pmb{0}_{n-1} \\
R & B & \pmb{0}_{n-1} \\
^{t} \pmb{x_n^*} & ^{t} \pmb{0}_{n-1} & -b w_n
\end{pmatrix}
,
\end{equation}
where 
\begin{gather*}
A \coloneqq -w_n {\rm diag} (a_1, a_2, \ldots, a_{n-1}), \quad B \coloneqq -b w_n I_{n-1}, \\
^{t} \pmb{x_n^*} \coloneqq -
\begin{pmatrix}
x_n^* & x_n^* & \cdots & x_n^*
\end{pmatrix}
, \\
a_j \coloneqq \frac{\alpha_j w_j^*}{w_n}, b \coloneqq \frac{(\sum_{l = 1}^{n-1} w_l^*) + w_n}{w_n},
\end{gather*}
$O_{m,n}$ denotes the $m \times n$ zero matrix, and $I_{n-1}$ is the identity matrix of order $n-1$. Note that $a_j$ and $b$ do not depend on $w_n$. Our purpose is to show that (i) the fixed point $(w_j^*,x_i^*)$ is asymptotically stable and (ii) damped oscillation occurs if $w_n$ is sufficiently small. In other words, we are going to prove the next theorem: 
\begin{theorem}
\label{theo_J}
(1) The real parts of the eigenvalues of the Jacobian matrix $J$ are all negative. \\
(2) There exists a positive constant $W$ such that $J$ has at least one imaginary eigenvalue if and only if $0 < w_n < W$. 
\end{theorem}
\begin{proof}
Let $f(z) \coloneqq \det (zI_{2n-1} - J)$ be the characteristic polynomial of $J$. Then, 
\begin{equation}
\label{dif_fz_det}
f(z) = (z+b w_n)
\det 
\begin{pmatrix}
z I_{n-1} & -A \\
-R & (z+b w_n) I_{n-1}
\end{pmatrix}
.
\end{equation}
Recall the next formula \cite{Serre2010}; if $S,T,U,V$ are square matrices and if $SU=US$, then
\begin{equation}
\det 
\begin{pmatrix}
S & T \\
U & V
\end{pmatrix}
= \det (SV-UT).
\label{linear}
\end{equation}
We have
\begin{widetext}
\begin{align}
& \det 
\begin{pmatrix}
z I_{n-1} & -A \\
-R & (z+b w_n) I_{n-1}
\end{pmatrix} \notag \\
= & \left| z(z+bw_n)I_{n-1} -
w_n
\begin{pmatrix}
a_1(x_1^* - q_1) & a_2 x_1^* &  \cdots & a_{n-2} x_1^* & a_{n-1} x_1^* \\
a_1x_2^* & a_2 (x_2^* - q_2) &  \cdots & a_{n-2} x_2^* & a_{n-1} x_2^* \\
\vdots & \vdots &  & \vdots & \vdots \\
a_1 x_{n-2}^* & a_2 x_{n-2}^* &  \cdots & a_{n-2} (x_{n-2}^* - q_{n-2}) & a_{n-1} x_{n-2}^* \\
a_1 x_{n-1}^* & a_2 x_{n-1}^* &  \cdots & a_{n-2} x_{n-1}^* & a_{n-1} (x_{n-1}^* - q_{n-1})
\end{pmatrix} 
\right| \notag \\
= & 
{\small
\begin{vmatrix}
z^2+bw_n z+w_n a_1(q_1 - x_1^*) & -w_n a_2 x_1^* &  \cdots & -w_n a_{n-2} x_1^* & -w_n a_{n-1} x_1^* \\
-w_n a_1x_2^* & z^2+bw_n z+w_n a_2 (q_2 - x_2^*) &  \cdots & -w_n a_{n-2} x_2^* & -w_n a_{n-1} x_2^* \\
\vdots & \vdots &  & \vdots & \vdots \\
-w_n a_1 x_{n-2}^* & -w_n a_2 x_{n-2}^* &  \cdots & z^2+bw_n z+w_n a_{n-2} (q_{n-2} - x_{n-2}^*) & -w_n a_{n-1} x_{n-2}^* \\
-w_n a_1 x_{n-1}^* & -w_n a_2 x_{n-1}^* &  \cdots & -w_n a_{n-2} x_{n-1}^* & z^2+bw_n z+w_n a_{n-1} (q_{n-1} - x_{n-1}^*)
\end{vmatrix} }
\notag \\
\notag
= &
\left(\prod_{k=1}^{n-1}a_{k}\right)
\begin{vmatrix}
\frac{z^2+bw_n z}{a_1}+w_n q_1 - w_n x_1^* & -w_n x_1^* &  \cdots & -w_n x_1^* & -w_n x_1^* \\
-w_n x_2^* & \frac{z^2+bw_n z}{a_2} + w_n q_2 - w_n x_2^* &  \cdots & -w_n x_2^* & -w_n x_2^* \\
\vdots & \vdots &  & \vdots & \vdots \\
-w_n x_{n-2}^* & -w_n x_{n-2}^* &  \cdots & \frac{z^2+bw_n z}{a_{n-2}}+w_n q_{n-2} - w_n x_{n-2}^* & -w_n x_{n-2}^* \\
-w_n x_{n-1}^* & -w_n x_{n-1}^* &  \cdots & -w_n x_{n-1}^* & \frac{z^2+bw_n z}{a_{n-1}}+ w_n q_{n-1} - w_n x_{n-1}^*
\end{vmatrix} 
.
\end{align}
\end{widetext}
By replacing $q_k \to (\frac{z^2+bw_n z}{a_k}+w_n q_k)$ and $x_k^{\dagger} \to w_n x_k^*$ in Lemma. \ref{Qn_calcu}, we see that 
\begin{widetext}
\begin{align*}
\det 
\begin{pmatrix}
z I_{n-1} & -A \\
-R & (z+b w_n) I_{n-1}
\end{pmatrix}
 &= \left(\prod_{k=1}^{n-1}a_{k}\right) \left(1- \sum_{l=1}^{n-1}  \frac{w_n x_l^*}{\frac{z^2+bw_n z}{a_l}+w_n q_l} \right) \prod_{k=1}^{n-1} \left(\frac{z^2+bw_n z}{a_k}+w_n q_k\right) \\
&= \left(1- \sum_{l=1}^{n-1}  \frac{a_l x_l^* w_n}{z^2+bw_n z+ a_l q_l w_n} \right)\prod_{k=1}^{n-1} \left(z^2+bw_n z +a_k q_k w_n \right).
\end{align*}
\end{widetext}
This implies that
\begin{multline}
\label{fz_supple}
f(z) = (z+bw_n) \left(1- \sum_{l=1}^{n-1}  \frac{a_l x_l^* w_n}{z^2+bw_n z+ a_l q_l w_n} \right) \\
\times \prod_{k=1}^{n-1} \left(z^2+bw_n z + a_k q_k w_n \right).
\end{multline}
We introduce a new variable $X$ and function $g(X)$ given as 
\begin{equation}
X \coloneqq \frac{z^2+bw_n z}{w_n}, \label{dif_X}
\end{equation}
and
\begin{equation}
g(X) \coloneqq \left(1- \sum_{l=1}^{n-1}  \frac{a_l x_l^*}{X+ a_l q_l} \right)\prod_{k=1}^{n-1} \left(X +a_k q_k\right). \label{dif_gX}
\end{equation}
It thus follows from Eq. \eqref{fz_supple} that 
\begin{equation}
\label{fz_gX}
f(z) = (w_n)^{n-1} (z+bw_n) g(X).
\end{equation} 
Note that $g(X)$ is a polynomial of degree $n-1$. 
\begin{lemma}
\label{negativity}
The solutions of $g(X) = 0$ are all negative real numbers.
\end{lemma}
\begin{proof}[Proof of Lemma. \ref{negativity}.]
We rewrite $g(X)$ as
\begin{equation}
\label{gX_rewrite}
g(X) = \left(1- \sum_{j=1}^{m}  \frac{\eta_j}{X + \delta_j} \right)\prod_{i=1}^{m} \left(X+\delta_i \right)^{\theta_i},
\end{equation}
where $\delta_i$, $\eta_i$, and $\theta_i$ satisfy the following conditions for $1 \leq i \leq m$: 
\begin{gather*}
\prod_{k=1}^{n-1} \left(X+a_k q_k \right) = \prod_{i=1}^{m} \left(X+\delta_i \right)^{\theta_i}, \\
\delta_k < \delta_l \quad \text{if} \quad k<l, 
\end{gather*}
and 
\begin{gather*}
\sum_{l=1}^{n-1}  \frac{a_l x_l^*}{X + a_l q_l} = \sum_{j=1}^{m} \frac{\eta_j}{X + \delta_j}. 
\end{gather*}
Note that $\delta_i > 0$, $\eta_i>0$, $\theta_i \in \N$, and $\sum_{i=1}^{m} \theta_i = n-1$. We introduce a new function $h(X)$, which is a polynomial of degree $m$, as
\begin{equation*}
h(X) \coloneqq \left(1- \sum_{j=1}^{m}  \frac{\eta_j}{X + \delta_j} \right)\prod_{i=1}^{m} \left(X+\delta_i \right).
\end{equation*}
It follows from Eq. \eqref{gX_rewrite} that  
\begin{equation}
\label{gX_hX}
g(X) = h(X) \prod_{i=1}^{m} \left(X+\delta_i \right)^{\theta_i - 1}.
\end{equation}
By using Eq. (\ref{dif_gX}) and the inequality (\ref{sumj_n-1}), we get
\begin{align*}
g(0) = \left(1- \sum_{l=1}^{n-1}  \frac{x_l^*}{q_l} \right)\prod_{k=1}^{n-1} \left(a_k q_k\right) > 0.
\end{align*}
Thus, $h(0)>0$ follows from Eq. \eqref{gX_hX}. We also have 
\begin{align*}
h(-\delta_i) = \left. -\eta_i \frac{\prod_{i=1}^{m} \left(X+\delta_i \right)}{X+\delta_i} \right|_{X=-\delta_i}.
\end{align*}
It follows from $\eta_i > 0$ that 
\begin{equation*}
\text{sgn}\ h(-\delta_i) = (-1)^{i}. 
\end{equation*}
We put $\delta_0 \coloneqq 0$. Then, by the intermediate value theorem, $h(X)$ has a root in each interval $(-\delta_i, -\delta_{i-1})$ for $1 \leq i \leq m$. Therefore, $h(X)$ has $m$ negative roots. From equation \eqref{gX_hX}, we have shown that the roots of $g(X)$ are all negative real numbers. 
\end{proof}
According to Lemma. \ref{negativity}, there exist positive real numbers $\xi_k$ $(1 \leq k \leq n-1)$ such that $\xi_1 \leq \xi_2 \leq \cdots \leq \xi_{n-1}$ and 
\begin{equation}
\label{gX_xi}
g(X) = \prod_{k=1}^{n-1} (X+\xi_k).
\end{equation}
Since $a_k$ and $x_l^*$ in Eq. \eqref{dif_gX} do not depend on $w_n$, $-\xi_k$ (i.e., the solutions of $g(X) = 0$) are independent of $w_n$. Substituting Eqs. (\ref{gX_xi}) and (\ref{dif_X}) into Eq. \eqref{fz_gX}, we have
\begin{equation}
f(z) = (z+bw_n) \prod_{k=1}^{n-1} (z^2+bw_n z+ \xi_k w_n).
\end{equation}
It follows from $bw_n>0$ and $\xi_k w_n > 0$ that the real parts of the solutions of $f(z) = 0$ are all negative. Thus, the fixed point $(w_j^*,x_i^*)$ is asymptotically stable. We also see that the Jacobian matrix $J$ has at least one imaginary eigenvalue if and only if
\begin{equation}
\label{nece_n}
0 < w_n < \frac{4 \xi_{n-1}}{b^2}.
\end{equation}
Hence, we have proved Theorem. \ref{theo_J}.
\end{proof}
Based on the same argument as in the previous section, the dynamics of slow variable $w_n$ is obtained as $\dot w_n= O(\eps^2)$, which implies that $w_n(t) =w_n(0)+O(\eps^2)$ for $t = O(1)$. Thus, from inequality (\ref{nece_n}), we conclude that oscillation occurs if 
\begin{equation}
v_n(0) = w_n(0)< \frac{4 \xi_{n-1}}{b^2}.
\end{equation}

In summary, assuming a few parameter conditions (i.e., the inequalities (\ref{til_x_necessary_sup}) and (\ref{sum_n})), we have shown that (i) all viral mutants have a shared effective growth rate $\Lambda=O(\eps)$ if one mutant has a considerably lower replication rate $\eps$ than the others and (ii) oscillatory viral dynamics occur if the initial value of the slowest-replicating mutant is sufficiently small. These findings are the generalization of those obtained in the previous section. 

\section{Discussion and conclusion}
By performing the linear stability analysis, we find that all growth rates can synchronize to the same value that is as small as the slowest growing agent (i.e., $O(\eps)$) in the previously proposed mathematical model \cite{nowak1992}. This explains the slow exponential growth observed in numerical simulations \cite{nowak1992}. We also determine that the oscillatory dynamics appear when the initial value of the slowest-growing agent is sufficiently small.


The inequality (\ref{sum_n}) represents the same result as in the previous studies \cite{nowak1992,nowak1994evolutionary} and it is the parameter condition in which the total viral load eventually diverges. Our study reveals that this condition is valid even without assuming fast dynamics of immune cells. We also find that the inequality (\ref{til_x_necessary_sup}) is the condition for the synchronization of all mutants' replication rates. Obviously, the following inequality 
\begin{equation}
 \sum_{l=1}^{n-1} \frac{1}{q_l}  < 1, \label{til_x_necessary_}
\end{equation}
is a sufficient condition for the inequality (\ref{til_x_necessary_sup}). Note that the parameter $1/q_i (\coloneqq \frac{r_i u}{p_i k})$ characterizes the strength of the $i$-th viral mutant compared with immunity. Then, the inequality (\ref{til_x_necessary_}) suggests that synchronized replication rates are observed when the total virulence of $v_1, \ldots , v_{n-2}$, and $v_{n-1}$ is insufficiently high to cause viral load divergence. 

The model we use in this paper was originally proposed as a model of human immunodeficiency virus (HIV) \cite{nowak1992}. 
Indeed, this model is consistent with a part of the virological features of HIV; e.g., HIV infects and destroys immune cells \cite{gorry2011coreceptors,turner1999structural}, HIV produces numerous mutants in the body \cite{cuevas2015extremely}, and various viral mutants may have different virulence \cite{rambaut2004causes}, which is represented by $r_i$ and $p_i$ in this model. 
However, this model has not been used to study HIV infection in recent years because the model is considered to be not biologically accurate to describe viral dynamics in vivo:  this model does not incorporate the uninfected target cell population, which is included in the standard HIV infection model \cite{perelson2013modeling}. 

In the later studies, Nowak and Bangham proposed another model that considers both the uninfected target cell and viral mutation \cite{nowak1996population} and Iwami et al. performed the linear stability analysis of this model for the case of one viral mutant with an assumption that viral dynamics are sufficiently fast \cite{iwami2006frequency}. Thus, applying our method to the Nowak \& Bangham model is a future work. 

There are several other possible extensions in our study. We transformed the nonautonomous system (\ref{n_model}) into the autonomous one (\ref{n_model_2}) by the approximation that $e^{\Lambda t} \simeq 1$ and assuming time-scale separation. These are appropriate approximations because as far as $t=O(1)$,  $e^{\Lambda t}$ and the slow variable $w_n$ stay $\eps$-vicinity of $1$ and $w_n(0)$, respectively. However, developing a more mathematically rigorous approach that is valid even when $t$ is sufficiently large, if any, is a future challenge. In addition, the global stability of the fixed point $(w_j^*,x_i^*)$ in the subsystem (\ref{n_model_2}) is an open problem in this study. As with the previous works that studied the global stability of a fixed point in mathematical models of infectious diseases \cite{kajiwara2015construction,wang2012global,li2012global}, constructing the Lyapunov function, if possible, is expected to solve this problem. 


In conclusion, we analyze in detail the simple mathematical model of asymmetrically interacting agents. We perform linear stability analysis and find the unique features of the model; i.e., the viral load initially oscillates and then slowly increases if the parameters and initial values satisfy a few conditions. Our work also proposes an analytical method of applying stability analysis to the exponentially diverging solution by using the techniques of variable transformation and time-scale separation. 

\section*{Conflict of Interest}
The authors have no conflicts to disclose.

\section*{Author's Contributions}
Y.K. initiated this research. H.K. proposed the research direction. Y.K. performed the analysis and numerical simulations. Y.K. wrote the manuscript under the mentorship of H.K.

\begin{acknowledgments}
This study was initiated by the first-named author at the summer school of iBMath (Institute for Biology and Mathematics of Dynamic Cellular Processes The University of Tokyo). We thank the former members of iBMath, especially Y. Nakata, H. Kurihara, and T. Tokihiro, for mathematical advice.
\end{acknowledgments}

\section*{Data Availability Statement}
The data that support the findings of this study are available within the article.

\appendix*
\section{The case when $n=3$}
We consider the following system:
\begin{subequations}
\label{n3_model}
\begin{align}
\dot v_1 &= v_1(1 - x_1) , \label{n3_model_v1}\\
\dot v_2 &= \alpha_2 v_2(1 - x_2) , \label{n3_model_v2}\\
\dot v_3 &= \alpha_3 v_3(1 - x_3) , \label{n3_model_v3}\\
\dot x_1 &= v_1\left[q_1 -  x_1(1+\frac{ v_2}{ v_1}+\frac{ v_3}{ v_1})\right] \label{n3_model_x1}, \\
\dot x_2 &=  v_2\left[q_2 -  x_2(1+\frac{ v_1}{ v_2}+\frac{ v_3}{ v_2})\right] \label{n3_model_x2}, \\
\dot x_3 &=  v_3\left[q_3 -  x_3(1+\frac{ v_1}{ v_3}+\frac{ v_2}{ v_3})\right] \label{n3_model_x3}.
\end{align}
\end{subequations}
We treat $\alpha_3$ as a small parameter (i.e., $\alpha_3 \ll 1$) and put $\eps \coloneqq \alpha_3$. The other parameters are assumed to be $O(1)$. As in the case of two viral mutants, we introduce new variables $\beta \coloneqq \frac{v_2}{v_1}$ and $\gamma \coloneqq \frac{v_3}{v_1}$. Assuming the convergence of $\beta$ and $\gamma$ to positive constants $\beta^*$ and $\gamma^*$, respectively, we see that 
\begin{multline}
   x_1^* = \frac{q_1}{1+\beta^*+\gamma^*}, \quad x_2^* = \frac{\beta^* q_2}{1+\beta^*+\gamma^*}, \\
   \quad {\rm and} \quad x_3^* = \frac{\gamma^* q_3}{1+\beta^*+\gamma^*},
\end{multline}
are fixed points of the subsystem given by Eqs. (\ref{n3_model_x1}), (\ref{n3_model_x2}), and (\ref{n3_model_x3}).
Substituting $x_i = x_i^*$ into the equations $\dot \beta = \frac{v_1 \dot v_2 - v_2 \dot v_1}{v_1^2} = 0$ and $\dot \gamma = \frac{v_1 \dot v_3 - v_3 \dot v_1}{v_1^2} = 0$, we obtain the following:
\begin{equation}
\label{beta*}
\beta^* = \frac{\alpha_2 q_1 + \eps (q_1 q_3-q_1-q_3) + \alpha_2 \eps q_3}{\alpha_2 q_2 + \eps q_3 + \alpha_2 \eps (q_2 q_3 - q_2 -q_3)} = \frac{q_1}{q_2} + O(\eps), 
\end{equation}
and 
\begin{equation}
\label{gamma*}
\gamma^* = \frac{\alpha_2 (q_1 q_2-q_1-q_2) + \eps q_1  + \alpha_2 \eps q_2}{\alpha_2 q_2 + \eps q_3 + \alpha_2 \eps (q_2 q_3 - q_2 -q_3)} = \frac{q_1q_2-q_1-q_2}{q_2} + O(\eps).
\end{equation}
For sufficiently small $\eps$, the conditions $\beta^*>0$ and $\gamma^*>0$ hold if
\begin{equation}
q_1 q_2 - q_1 - q_2 > 0. \label{q_3}
\end{equation}
We assume this inequality (\ref{q_3}) for the following argument. 

Substituting $x_i = x_i^*$ and $(\beta,\gamma) = (\beta^*,\gamma^*)$ into Eqs. (\ref{n3_model_v1}), (\ref{n3_model_v2}), and (\ref{n3_model_v3}), we acquire the following equation
\begin{align*}
 \dot v_i =  \lambda v_i, 
\end{align*}
where
\begin{align}
    \lambda = \frac{\alpha_2 \eps (q_1 q_2+q_2 q_3+q_3 q_1-q_1 q_2 q_3)}{\alpha_2 q_1 q_2+\eps q_3 q_1+\alpha_2 \eps q_2 q_3} = O(\eps). \label{lam_3}
\end{align}
We also assume $q_1 q_2+q_2 q_3+q_3 q_1-q_1 q_2 q_3 > 0$ so that $\lambda > 0$. 

Next we perform the linear stability analysis. By introducing new variables $w_i \coloneqq v_i e^{-\lambda t}$, we transform Eq. \eqref{n3_model} into the following six-dimensional nonautonomous system:
\begin{subequations}
\label{n3_model_2}
\begin{align}
\dot w_1 &= w_1(1 - \lambda - x_1), \\
\dot w_2 &= \alpha_2 w_2(1 - \frac{\lambda}{\alpha_2} - x_2), \\
\dot w_3 &= \eps w_3(1 - \frac{\lambda}{\eps} - x_3), \label{n3_model_2_w3}\\
\dot x_1 &= e^{\lambda t} w_1\left[q_1 -  x_1(1+\frac{ w_2}{ w_1}+\frac{ w_3}{ w_1})\right], \\
\dot x_2 &= e^{\lambda t} w_2\left[q_2 -  x_2(1+\frac{ w_1}{w_2}+\frac{ w_3}{ w_2})\right], \\
\dot x_3 &= e^{\lambda t} w_3\left[q_3 -  x_3(1+\frac{ w_1}{ w_3}+\frac{ w_2}{ w_3})\right]. 
\end{align}
\end{subequations}
As far as $t=O(1)$, we can safely replace $e^{\lambda t}$ in Eq. \eqref{n3_model_2} with $1$ since $e^{\lambda t} = 1 + O(\lambda t) = 1 + O(\eps)$. Moreover, $w_3$ is a slow variable and $w_1, w_2, x_1, x_2$, and $x_3$ are fast variables. Thus, in a good approximation, the dynamics of these fast variables are described by the five-dimensional autonomous subsystem as below: 
\begin{subequations}
\label{n3_model_s}
\begin{align}
\dot w_1 &= w_1(1 - \lambda - x_1), \\
\dot w_2 &= \alpha_2 w_2(1 - \frac{\lambda}{\alpha_2} - x_2), \\
\dot x_1 &= w_1\left[q_1 -  x_1(1+\frac{ w_2}{ w_1}+\frac{ w_3}{ w_1})\right], \\
\dot x_2 &= w_2\left[q_2 -  x_2(1+\frac{ w_1}{w_2}+\frac{ w_3}{ w_2})\right], \\
\dot x_3 &= w_3\left[q_3 -  x_3(1+\frac{ w_1}{ w_3}+\frac{ w_2}{ w_3})\right],
\end{align}
\end{subequations}
in which $w_3$ is regarded as a constant. The internal fixed point of this subsystem (\ref{n3_model_s}) is 
\begin{multline}
(w_1, w_2, x_1, x_2, x_3) 
= \\
\left( \frac{w_3}{\gamma^*},\frac{\beta^* w_3}{\gamma^*},\frac{q_1}{1+\beta^* +\gamma^*},\frac{\beta^* q_2}{1+\beta^* +\gamma^*},\frac{\gamma^* q_3}{1+\beta^* +\gamma^*} \right).  \label{fix_3}
\end{multline}
The Jacobian matrix at this fixed point is
\begin{equation}
\label{jacobi_3_m}
\begin{pmatrix}
0 & 0 & -e_1 & 0 & 0 \\
0 & 0 & 0 & -e_2 & 0 \\
e_3 & -e_4 & -e_8 & 0 & 0 \\
-e_5 & e_6 & 0 & -e_8 & 0 \\
-e_7 & -e_7 & 0 & 0& -e_8
\end{pmatrix}
, 
\end{equation}
where $e_1, \ldots, e_8$ are positive constants given by
\begin{equation}
\label{e_dif}
\begin{split}
& e_1 = \frac{w_3}{\gamma^*},\quad e_2 = \frac{\alpha_2 \beta^* w_3}{\gamma^*},\quad e_3 = \frac{q_1 (\beta^* + \gamma^*)}{1+\beta^* + \gamma^*}, \\
& e_4 = \frac{q_1}{1+\beta^* + \gamma^*}, \quad e_5 = \frac{q_2 \beta^* }{1+\beta^* + \gamma^*}, \quad e_6 = \frac{q_2(1+\gamma^*)}{1+\beta^* + \gamma^*}, \\
& e_7 = \frac{q_3 \gamma^*}{1+\beta^* + \gamma^*},\quad \mathrm{and} \quad e_8 = \frac{(1+\beta^* + \gamma^*)w_3}{\gamma^*}.
\end{split}
\end{equation}
The eigenvalues of matrix (\ref{jacobi_3_m}) are
\begin{widetext}
\begin{multline}
\label{n3_eigen}
 -e_8, \quad -\frac{e_8}{2}+\frac{1}{2} \sqrt{e_8^2 - 2(e_1 e_3+e_2e_6)\pm 2\sqrt{(e_1 e_3-e_2 e_6)^2+4e_1 e_2  e_4 e_5}},\\
{\rm and} \quad  -\frac{e_8}{2}-\frac{1}{2} \sqrt{e_8^2 - 2(e_1 e_3+e_2e_6)\pm 2\sqrt{(e_1 e_3-e_2 e_6)^2+4e_1 e_2  e_4 e_5}}.
\end{multline}
\end{widetext}
From $e_3 e_6 > e_4 e_5$, we conclude that the real parts of these eigenvalues are all negative, which implies that the fixed point (\ref{fix_3}) is asymptotically stable. Moreover, imaginary eigenvalues appear if
\begin{align}
\label{e_ine}
e_8^2 - 2(e_1 e_3+e_2e_6) - 2\sqrt{(e_1 e_3-e_2 e_6)^2+4e_1 e_2  e_4 e_5} < 0. 
\end{align}
Substituting Eq. \eqref{e_dif} into inequality (\ref{e_ine}), we have
\begin{widetext}
\begin{align}
& w_3^2 \left(\frac{1+\beta^* + \gamma^*}{ \gamma^*}\right)^2 -2w_3 \left\{  \frac{q_1 (\beta^* + \gamma^*) + \alpha_2 q_2 \beta^* (1+\gamma^*) + \sqrt{[q_1 (\beta^* + \gamma^*) - \alpha_2 q_2 \beta^* (1+\gamma^*)]^2+4\alpha_2 q_1 q_2 (\beta^*)^2}}{\gamma^*(1+\beta^* + \gamma^*)}  \right\} < 0 \notag \\
& \iff 0< w_3 < \frac{2 \gamma^* \left\{ q_1 (\beta^* + \gamma^*) + \alpha_2 q_2 \beta^* (1+\gamma^*) + \sqrt{[q_1 (\beta^* + \gamma^*) - \alpha_2 q_2 \beta^* (1+\gamma^*)]^2+4\alpha_2 q_1 q_2 (\beta^*)^2} \right\} }{(1+\beta^* + \gamma^*)^3}. 
\label{w3condition}
\end{align}
\end{widetext}
\normalsize
The fast variables stay in the $\eps$-vicinity of the fixed point in the full system after the transient process because the subsystem of the fast variables has the stable fixed point (\ref{fix_3}). Substituting $x_3=x_3^*+O(\eps)$ into Eq. \eqref{n3_model_2_w3} and further using Eqs. (\ref{beta*}), (\ref{gamma*}), and (\ref{lam_3}), we obtain $\dot w_3= O(\eps^2)$, which implies that $w_3(t) = w_3(0)+O(\eps^2)$ for $t = O(1)$. Therefore, $w_3$ in inequality (\ref{w3condition}) can be regarded as $w_3(0)$ in a good approximation. Thus, we conclude that oscillatory viral dynamics occur if there is a viral strain whose replication rate is considerably lower than the others and its initial value is sufficiently small. We also find that all viral mutants have the following shared effective growth rate: $\lambda = O(\eps)$. Note that these findings are the same as those obtained in the case of two viral mutants.

\bibliography{HIV}

\end{document}